\documentclass[10pt,a4paper]{article}

\usepackage{amsmath,amssymb,amsthm}
\usepackage{url}
\usepackage{placeins}
\usepackage{subfigure}
\usepackage{multirow}
\usepackage{epsfig}
\newcommand{\Pbb}{\ensuremath{\mathbb{P}}}
\newcommand{\PbbPI}{\ensuremath{\mathbb{P}^{\rm PI}}}
\newcommand{\FbbPI}{\ensuremath{\mathcal{F}^{\rm PI}}}

\newcommand{\Q}{\ensuremath{\mathbb{Q}}}
\newcommand{\E}{\ensuremath{\mathbb{E}}}

\newcommand{\CONSTANT}{{{\it{Constant}}}}

\newcommand{\rolls}{\ensuremath{{\rm {rolls}}}}
\newcommand{\nrolls}{{\ensuremath{n_\rolls}}}
\newcommand{\cav}{\ensuremath{{C^{\rm av}}}}
\newtheorem{myProp}{Proposition}

\newcommand\be{$$}
\newcommand\ee{$$}
\newcommand\ben{\begin{equation}}
\newcommand\een{\end{equation}}
\newcommand\bea{\begin{eqnarray*}}
\newcommand\eea{\end{eqnarray*}}
\newcommand\bean{\begin{eqnarray}}
\newcommand\eean{\end{eqnarray}}

\author{Chris Kenyon and Andrew Green\footnote{Contact: chris.kenyon@lloydsbanking.com}}
\title{Regulatory-Optimal Funding\footnote{\bf The views expressed are those of the authors only, no other representation should be attributed.}
}
\date{17 January 2014\\ \vskip5mm Version 1.20}

\begin{document}

\maketitle

\begin{abstract}
Funding is a cost to trading desks that they see as an exogenous input.  Current FVA-related literature reflects this by also taking funding costs as an exogenous input, usually constant, and always risk-neutral. 
However, this funding curve is the output from the point of view of the Treasury department of a bank.  The Treasury must consider Regulatory-required liquidity buffers, and both risk-neutral (\Q) and physical measures (\Pbb).  We describe the Treasury funding problem and optimize against both measures, using the Regulatory requirement as a constraint.  We develop theoretically optimal strategies for \Q\ and \Pbb\ measures, and then demonstrate a combined approach in four markets (USD, JPY, EUR, GBP).  Since we deal with physical measures we develop appropriate statistical tests, and demonstrate highly significant (p$<$0.00001), out-of-sample, improvements on hedged funding with a combined approach achieving 44\%\ to 71\%\ of a perfect information criterion.  Thus regulatory liquidity requirements change both the funding problem and funding costs.
\end{abstract}

\section{Introduction}

The costs to trading desks of holding positions typically include funding, and at least five banks\footnote{Barclays, Goldman Sachs, Lloyds Banking Group, Royal Bank of Scotland, JP Morgan (for the last one, see Financial Times 15 Jan 2014).} \cite{Cameron2013a} report funding value adjustments (FVA).  Current FVA-related literature treats funding costs as an input \cite{Burgard2011a,Burgard2012a,Morini2011a,Crepey2012a,Crepey2012b,Pallavicini2011b,Pallavicini2012a} that is usually constant (\cite{Pallavicini2013a} is an exception), and always risk-neutral.  Books specifically on liquidity \cite{Castagna2013a} or from a Treasury point of view \cite{Choudhry2012a}  do not treat funding costs as constant, but do not cover funding optimization mathematically.  We take the Treasury point of view for whom the funding curve is an output based on a funding strategy.  Treasury must consider Regulatory requirements \cite{FSA-P0916} (liquidity buffers) and both risk-neutral (\Q) and physical measures (\Pbb).  Precisely because Regulatory-required funding buffers exist, trading desks will always see funding costs that are not equal to the riskless rate when they have funding requirements.  This is true even if the bank can fund itself, unsecured, at the riskless rate.

The objective of this paper is to describe the Treasury funding problem with respect to Regulatory-required liquidity buffers and provide an optimal solution taking into account both the risk-neutral measure (\Q) and the physical measure (\Pbb).  Our contributions are: firstly, the description of the Regulatory-Optimal Funding problem; secondly, the provision of theoretically optimal funding strategies for \Q\ and \Pbb\  solving the Regulatory-Optimal Funding problem;  thirdly the development of practical optimization strategies involving both \Q\ and \Pbb; fourthly the calibration of the practical strategies and the demonstration of their effectiveness in four markets (USD, JPY, EUR, GBP), from 1995\footnote{From 2000 for EUR.} to 2012.  Since we deal with physical measures, a further contribution is the development of appropriate measures, and statistical tests, for testing strategy performance on out-of-sample data.  We demonstrate improvements on hedged funding with a combined approach achieving up to 71\%\ of a perfect information criterion (perfect information is when the future is known).

We detail our assumptions in the modeling section, however we highlight the most significant here.
\begin{itemize}
	\item Practically, we focus on short-term funding out to one year.  This is generally the area for money-market activity, rather than bond issuance and asset-liability management.  
	\item Theoretically our approach applies to any maturity but funding buffers are relatively less important for longer funding.  Despite this, the attractiveness of shorter --- and usually cheaper --- funding remains when yield curves are upwards-sloping.  Funding buffers are required by Regulators exactly because of this attractiveness to provide a degree of systematic resilience that might otherwise not be present.  Regulatory details are in the next section.  
	\item We assume that a risk-neutral measure (\Q) exists.  The significance of this assumption is that in the risk-neutral measure future expected yield curves can be derived from the present yield curve.  In contrast we define physical measures (\Pbb), and physical filtrations, as any other measures or filtrations that are not the risk neutral one.  The risk neutral measure (and filtration) is the one that can be derived from current market prices (see standard texts for details, e.g. \cite{Shreve2004a}).  A physical measure, for example involving the prediction of future expected yield curves, can be derived from any source but has the distinction that prices derived from it cannot be hedged.  A trader's view of the market is one example of a physical measure. This view may lead the trader to place open positions that will create a profit if the view turns out to be correct.  Section \ref{s:Pid} specifies the identification of the physical measure in this paper.
	\item We take Xibor as a proxy for funding costs simply as a place to start and because funding costs are often quoted as a spread over Xibor (the X stands for the relevant currency).  Xibor is an average of interbank lending rates, typically available up to 12 months or 18 months.  Because it is an average, individual contributing banks will actually experience a spread over or under Xibor for their lending with these horizons.   
	
	Xibor is an unsecured borrowing rate.  As such it already includes the possibility that the borrower may default, i.e. the credit risk.  We optimize the bank's funding strategy assuming that it is a going concern.  We do not consider the possibility that it may default on its liabilities as part of this strategy.  If we did, then --- conditional on survival --- the bank would leak PnL.  
	\item We assume no volume effects, i.e. no market impact of funding requests on funding costs.  Clearly there are volume effects, and these can be integrated into our framework, but we leave them for future research. 
	\item Our optimization involves expectation-based utility functions.  More complex utility functions could be used, for example involving VAR or Expected Shortfall \cite{BCBS-219}, we leave these extensions for future work. 
\end{itemize}

Bank funding overall is complex with multiple sources (deposits, debt markets, money markets, etc) and governed by a liquidity funding policy.  Any liquidity funding policy must be executed with respect to different sources.  In as much as markets are involved all liquidity funding policies are embodied by trading strategies.  The overall liquidity funding policy will involve asset-liability planning as well.  We focus on one area, funding optimization of the regulatory-required liquidity buffer.  All banks must solve this problem (provision and cost of the buffer), even if they do so sub-optimally.  Note that banks are not passive price-takers of funding costs over multi-year horizons.  Executives chose a risk-appetite which they must then execute on.  The success of this risk-appetite strategy will influence bank funding costs from the market.  Given our timescale (up to one year) we do not include these effects.  Furthermore bankwide funding cannot be considered in isolation from capital and credit, as the emergence of xVA desks \cite{Carver2013a,Carver2013b} demonstrates.

Banks direct internal funding using Funds Transfer Pricing (FTP) which can have local and global features.  For example a local feature could be relatively cheap/expensive funding at a given maturity to correct an observed Asset-Liability mismatch.  A global feature could involve relatively cheap/expensive funding for a particular business line.  However the explicit and  detailed control for business objectives is poorly developed \cite{fsa:dt}.  Despite this, differential funding for different desks is typically present, e.g. repo desks versus OTC derivative desks.

We start by recalling the discussion in \cite{Hull2011a} (rather than \cite{Hull2012a}) on hedging input costs which we apply to funding.  Let us assume that different banks have different (funding) costs and that Bank A hedges its (funding) costs.  Now suppose that the general market climate improves and the systematic part of funding costs decreases for all banks.  The Treasurer of Bank A must now explain to the CEO and the Board why Bank A is losing money relative to its competitors when offering similar prices for similar products.  Thus (funding) hedging decisions should always be at Board level.  Our approach avoids such issues by combining \Q\ and \Pbb.

\subsection{Regulation}

The FSA scenario \cite{FSA-P0916} (now owned by the PRA) includes two weeks of market closure followed by an individually agreed level of market access for the remainder of the three months (the Individual Liquidity Guidance, ILG).  Effective funding stress can be much larger than reductions in wholesale funding when deposit leakage is included, depending on the type of bank (retail, commercial, or investment).

Basel III \cite{BCBS-188,BCBS-189} by contrast bases its Liquidity Coverage Ratio (LCR) on total net cash outflows over 30 calendar days.  This is defined (paragraph 50) as outflows minus the minimum of inflows and 75\%\ of outflows.  The document specifies the runoff factor to be applied to each category of funding, e.g. "Unsecured wholesale funding provided by non-financial corporates and sovereigns, central banks and public sector entities: 75\%" and "Unsecured wholesale funding provided by other legal entity customers: 100\%" which includes non-central banks and financial firms.  Thus wholesale funding will vanish under the Basel III LCR calculation.

Prospective regulations on Prudent Valuation \cite{EBA-CP-2013-28}, which is clarified in \cite{EBA-CP-2013-28-FAQ}, propose that for capital purposes funding should be calculated maturity matched with trade payments:  
\begin{quote}
Art 10 ``... what adjustment must we make for the FVA'' reply: ``FVA should be assessed using the wholesale funding costs of the firm, maturity matched with the contractual trade payments.''
\end{quote}
We note that this does not mean the funding must be available --- only that for capital purposes pricing should include funding costs as though funding were available.  Whilst maturity matching payments may be straightforward for zero coupon bonds its direct application to derivatives is problematic.  An uncollateralized vanilla interest rate swap traded back-to-back with the street may require funding in the future or it may provide funding depending on market movements.  Hence we see that the precise description of future funding for derivatives can be characterized by four dimensions as a minimum: start date; tenor; rate; and volume.  We might say that this data prices volume-dependent funding swaptions.    Most funding literature, except \cite{Pallavicini2013a}, is based on a one-dimensional curve where start date is the present and only the tenor is present.  In any case pricing using this funding data hypercube is not practical because it cannot be calibrated.  This paper offers a pragmatic alternative approach.  

\subsection{Previous Work}

The key papers on funding consider either replication \cite{Burgard2011a,Burgard2012a} and / or (effectively) project finance \cite{Morini2011a,Hull2012a}.  In \cite{Burgard2011a,Burgard2012a} the authors assume that trading in own-bonds is possible and has no effect on the funding costs of the overall firm.  \cite{Hull2012a} assumes that firm funding costs exactly reflect the riskiness of the set of projects that the firm undertakes.    Thus marginal funding is the same as stand-alone funding.  In addition they assume that DVA can be divided into firm-default benefit and funding-default benefit, the latter of which exactly equals funding costs.  This contrasts with the hedging setup by the same author described in the introduction.  \cite{Morini2011a}  shows that this depends on not assuming that the firm survives.  \cite{Hull2013a} includes this point of view and comes to similar conclusions. \cite{Castagna2012a} also describe the conditions under which DVA can be replicated as well as interactions between DVA and FVA.  If the analysis is done conditioning on firm survival then funding-default benefit no longer equals funding costs because only avoided-funding cancels with funding costs.   Net funding at any point pays full funding costs.  It is this net funding that we address.

Books specifically on liquidity or from a Treasury point of view \cite{Castagna2013a,Choudhry2012a}  do not treat funding costs as constant.  However, they do not cover funding optimization mathematically as here.

Prediction of future market conditions \cite{Asness2013a,Geczy2013a}, and by yield curves \cite{Estrella2006a}, has a long history and a large literature based on time-series analysis \cite{Hamilton1994a}.  Our aim is not to contribute to this literature but to use the simplest possible approach (momentum, aka exponentially-weighted moving average or EWMA) to provide a \Pbb-linked strategy that is demonstrably better than hedging alone.

\section{Modeling}

We aim to minimize funding costs by choosing an optimal funding roll ($\alpha$).  The bank must always have a buffer of at least $\Delta$ days of funding.  Although our setup is general our focus is medium-term funding, i.e. out to one year, in four markets (USD, JPY, EUR, GBP).  We first analyze the characteristics of optimal \Q\ and optimal \Pbb\ strategies under simple assumptions, and then turn to the problem of \Pbb-measure discovery.  We then apply standard techniques to calibration and out-of-sample model analysis.

We assume that the bank can fund spot according to a known curve that we will call the yield curve.  In our examples we will use the Libor curve as a proxy for a typical bank as it contains credit and liquidity elements.  As soon as funding is shorter than the regulatory buffer, $\Delta$, it is worthless to the bank.  Thus we include a bid-ask spread $\varphi$ on the funding curve.  $\varphi$ gives the fraction of the funding cost recoverable by selling funding.  We will show that the bid-ask spread is an important driver of optimal funding roll characteristics.
\be
\varphi := L_{\rm bid}(*) / L_{\rm ask}(*) 
\ee
Where we use $*$ to indicate that the bid-ask spread is not tenor-dependent in our model.  $L(*)$ is the input funding curve.  We do not write ``bid'' and ``ask'' explicitly but use the combination of the sign of the flow and $\varphi$ to express the effects.  The model only uses the bid-ask spread at maturity $\Delta$ which is when held funding no longer counts towards the liquidity buffer.

Our base funding cost setup is \Q-funding, i.e. as assumption that the bank can fund in the future at rates predicted from the current yield curve.  However, we also assume that only spot funding is practical.  Thus in historical testing subsequent funding rolls must be carried out with the previously defined roll up to the horizon.  For \Pbb-funding spot funding also comes from the spot yield curve but future funding with predicted yield curves is included in the optimization.  Practically all funding strategies will re-optimize at every funding roll, we leave this for future research using multi-stage stochastic optimization \cite{Birge2011a}.  Here we explore the limits of myopic strategies, i.e. decide the funding roll at the start and then apply it up to the horizon.

The current work is to demonstrate that an assumption of hedged funding is sub-optimal in practice, and that using predicted funding curves is possible.  By possible we mean that out-of-sample statistical, and significant practical improvements, are observed.  We compare our optimal strategies using predictions with what we would do with perfect information.  This is possible from historical analysis.  Thus we report results in terms of average bps improvement over an attempted hedging strategy, and as a percentage achievement compared to perfect information.

\subsection{Setup}
Our basic setup is as follows (Figure \ref{f:figure_roll}).
\begin{itemize}
	\item Fixed quantity of funding up to horizon $h$.  
	\item Regulatory funding buffer requirement as a constraint that the bank must always have sufficient funding for $\Delta$ days.
	\item Input funding curve, i.e. where primary issuance can be sold for different maturities, $L(*)$.  (This may well be different from secondary trading prices.)
\end{itemize}
The existence of a funding buffer means that funding rolls must overlap.  Otherwise as we get towards the end of a funding roll, i.e. less than $\Delta$ days away, the regulatory constraint is broken. 

\begin{figure}[htbp]
	\centering
		\includegraphics[width=0.90\textwidth,clip,trim=5 25 5 25]{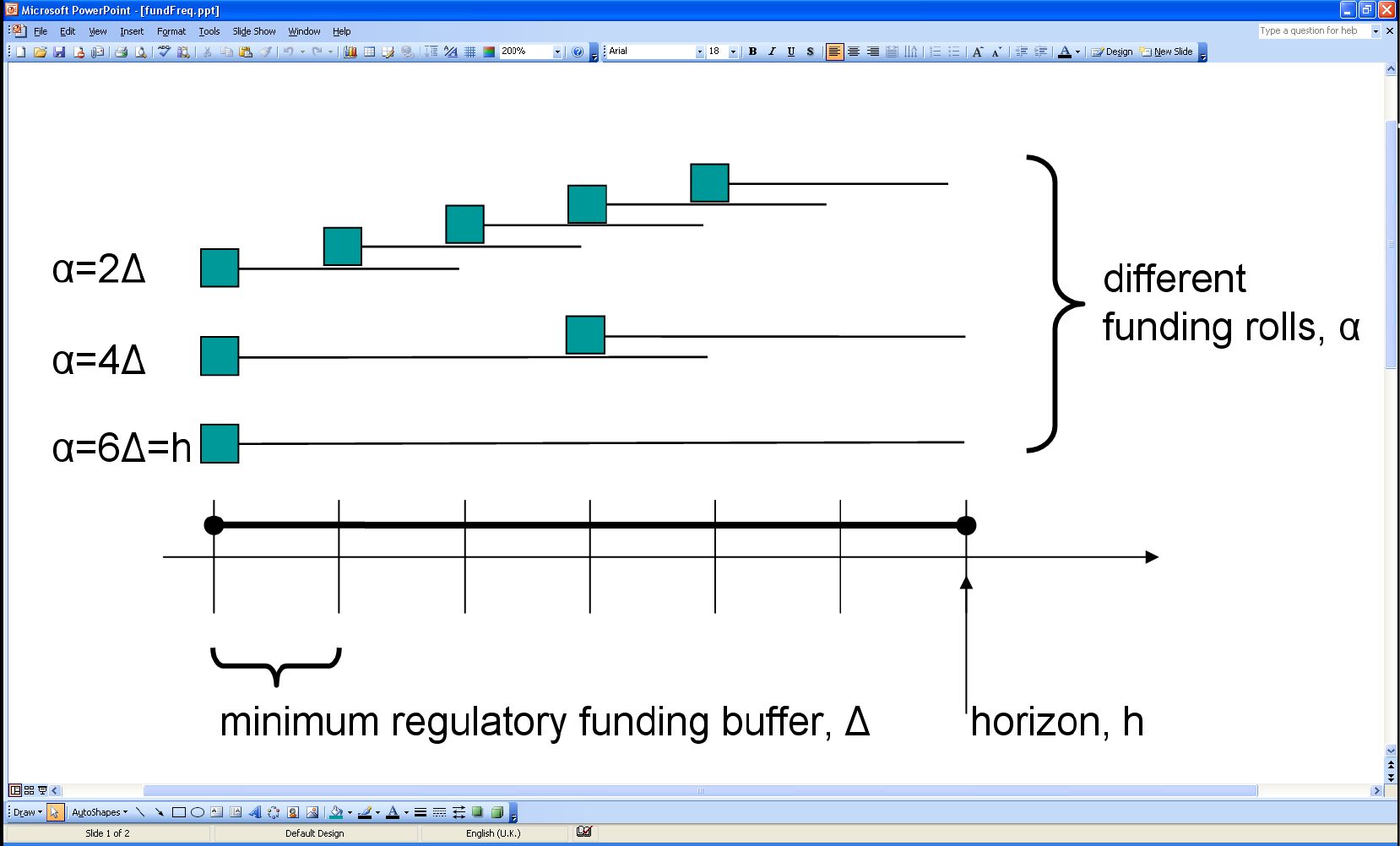}
	\caption{Basic funding setup showing different funding rolls $\alpha$ obeying the minimum funding buffer time $\Delta$, up to the funding horizon $h$.  Squares show start of each funding roll.}
	\label{f:figure_roll}
\end{figure}

The first point in the setup deals with the question of how to define the average funding cost.  We do the average for a fixed horizon.  There are many possible definitions addressing different combinations of getting an average cost over a period versus the repeating nature of the problem.  We have picked this one as corresponding to a planning horizon.  The actual horizon is important because we have an input  funding curve which may be upward sloping, flat, or possibly downward sloping, with variable curvature.

The maturity of the funding need sets a limit on the proportion of the time we must have a funding overlap.  For a very long maturity with term funding the buffer overlap will be a very small part of the total funding time.  This also illustrates that term funding is sub-optimal for short instruments assuming that the relevant desk has a stable business model.  If the desk has a highly volatile business model then it will suffer term funding with its attendant buffer costs.  Whilst the Treasury sees the funding requirements of the bank, not just each desk, there can still be volatility in funding needs in as much as desks (and business lines) have systematic dependencies.  

We consider a fixed funding horizon requirement $h$ and fixed roll length $\alpha$.  Any desk (or bank) will have a variety of funding maturities required at any time but we start from separate requirements.  Regulatory-optimal funding in the face of uncertain requirements is outside the scope of this paper and is an area for future research.

The overlap of roll periods highlights the significance of the bid-ask spread for funding.  Once we start the next roll, the previous funding with length $\le\Delta$ is of no use and can be sold.  The difference with the input funding curve may well be significant and we capture this with a parameter $\varphi$.  Generally $\varphi$ will depend on $\alpha$ to incorporate the market impact of selling but we use a constant for simplicity.  Note that whatever the funding roll the quantity of funding will always be constant.  This is because we assume that unnecessary funding will be sold immediately.  Any funding with a maturity less than the buffer period ($\Delta)$ does not count towards the buffer (by definition).  It does not matter what the bank does with these assets from the point of view of the Regulator, since they do not count for the liquidity buffer.  However, for the bank, it does matter in terms of value.  Since we assume that the interbank market pays Xibor we start from Xibor for how they can be invested

If the bid-ask spread is zero then it may appear that the optimal strategy, with an upwards-sloping yield curve, is to set $\alpha=\Delta+1$day.  Thus the bank would buy and sell funding every day.  Although this strategy appears to enable the bank to obtain funding at the overnight rate this is not optimal.  Consider the case where interest rates are going up faster than the yield curve.  Under these circumstances (lagging yield curve) it is optimal to get funding for longer (we are back to our hedging-or-not discussion).  Of course this assumes that the bank can detect such situations.  Generally the bid-ask spread will not be zero and we include it in our calculations.

\FloatBarrier
Working in units of $\Delta$ the number of times funding must be rolled to horizon $h_n$ with roll length $\alpha_n$, $n_\rolls(h_n,\alpha_n)$ is:
\ben
\nrolls(h_n,\alpha_n) = 
\begin{cases}
0 & \alpha_n \ge h_n \\
\left\lceil \frac{h_n-\alpha_n}{\alpha_n-1} \right\rceil & {\rm otherwise}\\
\end{cases}
\label{e:nroll}
\een
where: $\alpha_n=\alpha/\Delta$; $h_n = h/\Delta$.  Thus the gross excess funding  $G_{e\!f}$ will be, as a percentage:
\ben
G_{e\!f} = 100\times\frac{n_\rolls}{h_n}
\label{e:gross}
\een
Gross excess funding,  $G_{e\!f}$, is funding with maturity less than $\Delta$.  This funding does not count for the Regulatory buffer, hence it is excess funding.   We call it gross because it is simply a quantity not the cost of the quantity.

Figure \ref{f:figure_HR} shows Equations \ref{e:nroll}, \ref{e:gross} from the point of view of constant roll lengths (left panels) and from the point of view of constant horizons and choice of roll length (right panels).  All lengths are in units of $\Delta$.  The advantage of moving to longer rolls lengths is clear in terms of gross excess funding.  However, longer roll lengths mean more expensive funding when the input funding curve is upwards sloping, in the absence of other factors.  

\begin{figure}[htbp]
	\centering
		\includegraphics[width=0.90\textwidth,clip,trim=0 0 0 0]{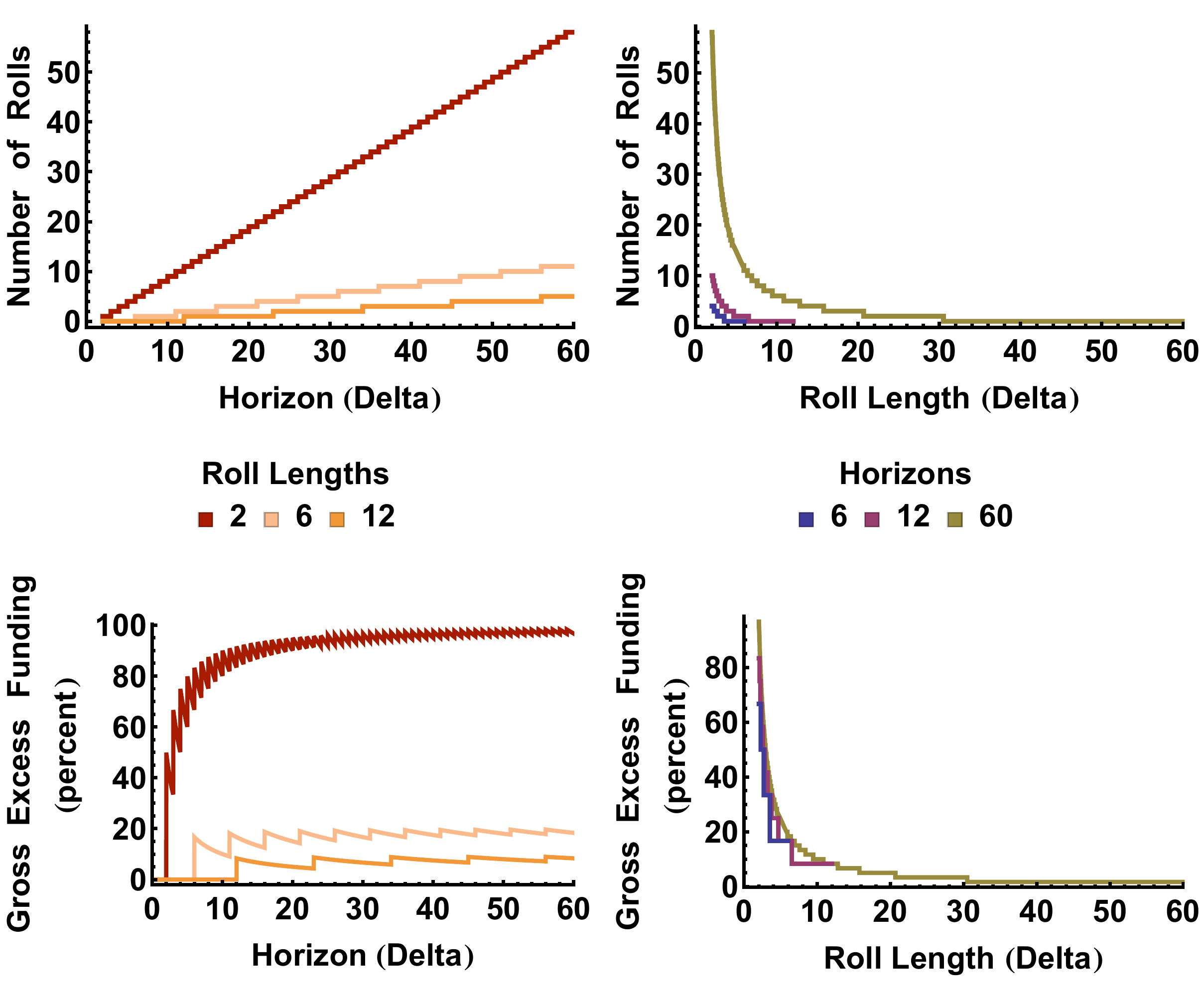}
	\caption{Equations \ref{e:nroll}, \ref{e:gross} from the point of view of constant roll lengths (left panels) and from the point of view of constant horizons and choice of roll length (right panels).  All lengths are in units of $\Delta$.}
	\label{f:figure_HR}
\end{figure}

\FloatBarrier
\subsection{Projected Costs and Optimization}

We use the expected undiscounted funding cost $\cav$ as our primary metric because this is often used in practice.  It is often quoted internally in banks in terms of the number of basis points above some reference curve (possibly an overnight curve, e.g. Fed Funds, Eonia, Sonia, etc).  This may be calculated under a physical (\Pbb) or risk-neutral (\Q) measure depending on the circumstances.  This mostly refers to whether funding costs are hedged or not.  As discussed above, we use Xibor is a proxy rate for unsecured borrowing. 

Considering Figure \ref{f:figure_roll} we may write $\cav$ as:
\begin{equation}
\begin{aligned}
\cav_t^{\ddag}&(\alpha,h)=\\
&\frac{1}{h}\E_t^\ddag\left[  
\min(\alpha,h) F_t(t,t+\min(\alpha,h) )     \vphantom{\sum_{i=1}^{\nrolls} }\right. \\
&\ \ \qquad{} +
\sum_{i=1}^{\nrolls} 
\left(  \vphantom{\sum}
-\varphi\Delta F_t(t+i(\alpha-\Delta),t+\Delta+i(\alpha-\Delta))   \right.  \label{e:cav}\\
&\ \ \quad\quad\qquad\qquad   {}+  \min(\alpha,h-i(\alpha-\Delta)) F_t(t+i(\alpha-\Delta),t+\min(i(\alpha-\Delta)+\alpha,h)
)
\left.\left.   \vphantom{\sum}\right)\vphantom{\sum_{i=1}^{\nrolls}}\right]
\end{aligned}
\phantom{\hspace{3cm}}
\end{equation} 
Where: 

$F_t(t_1,t_2)$ is the forward rate as seen from $t$ between $t_1$ and $t_2$; 

if $\nrolls<1$ then there are no terms in the summation;

$\ddag$ measure used in the expectation (\Pbb, \Q, or some combination).

Funding costs depend on the interpretation of $F_t(t_1,t_2)$, i.e. what is its actual realizable value, or what value can be locked in.  We work with continuously compounded rates for simplicity.

The myopic version of our optimization problem based on Equation \ref{e:cav} is:
\ben
\cav_{t,{\rm opt}}^{\ddag}(h) = \min_\alpha{\cav_t^{\ddag}(\alpha,h)}   \label{e:opt}
\een
We call this problem myopic as we are only permitted to chose $\alpha$ once, at the start.  This is not as great a limitation as it may appear when we consider that the funding we are optimizing is primarily short term, i.e. within O(10)$\Delta$.  In addition, practically, to avoid market impact banks typically stagger their funding so as to be in the market every day with some fraction.

In general Equation \ref{e:opt} is a non-linear, non-convex, optimization problem because of the minimum terms (otherwise for polynomial yield curves it would be an equivalent polynomial problem).  We assume that the yield curve is linear over the first year.  The linear yield curve assumption is to enable a basic theoretical understanding of the nature of the optimal funding problem.  For a linear model, the mean r-squared values for the four currencies over the period considered (weekly samples, Xibor maturity points, continuous compounding) are \{BP=0.66, EU=0.87, JP=0.76, US=0.80\}, and the geometric mean of the p-values\footnote{A geometric mean for p-values is appropriate because they cover a wide range, essentially log-scaling.} are \{BP=4.5e-6, EU=2.8e-7, JP=5.6e-8, US=1.1e-7\}, thus a linear model is supported by the data.  We are not making a formal hypothesis test here, because we know that reality is more detailed.  If we were making a hypothesis test then we would include Bonferroni (or similar) corrections for multiple tests and would phrase our observation as ``not rejected at a chosen p-level''.

If we assume a linear (continuously compounding) yield curve $y(T)$ with constants $a,b$:
\be
y(T) = a + b T
\ee
then, WLOG\footnote{WLOG = without loss of generality.} setting $t=0$ for simplicity, Equation \ref{e:cav} becomes:
\begin{equation}
\begin{aligned}
\cav_0&(\alpha,h)=\\
&\ \ \frac{1}{h}\left[  
(a+b\min(\alpha,h))\min(\alpha,h)      \vphantom{\sum_{i=1}^{\nrolls} }\right. \\
&\ \ \qquad{} +
\sum_{i=1}^{\nrolls} 
\left(  \vphantom{\sum}
-\varphi  \Delta[a + 2b i (\alpha-\Delta)+\Delta b  ]      \right.  \label{e:cavL}\\
&\ \ \quad\quad\qquad\qquad   {}+  q_i[a + 2b i (\alpha-\Delta)+q_i b  ] 
)
\left.\left.   \vphantom{\sum}\right)\vphantom{\sum_{i=1}^{\nrolls}}\right]
\end{aligned}
\phantom{\hspace{3cm}}
\end{equation} 
where:

$q_i$ is $\alpha$ except for $i=\nrolls$ when it equals $\min(\alpha, h - i (\alpha-\Delta))$

\subsubsection{Optimal \Q\ Funding}

\Q\ Funding is where we solve the optimization problem under the \Q\ measure.  This measure is determined by the observed prices of tradeable instruments at $t=0$ (see any text for details, e.g. \cite{Shreve2004a}).  For example in the \Q\ measure the current yield curve perfectly predicts the expected future yield curves.

It may be very difficult to lock in forward funding according to an observed input yield curve.  However, for completeness we include this case.  
\begin{myProp}
Given a linear input yield curve $(a>0,b)$ the regulatory-optimal \Q-funding strategies with horizon $h$ are:
\begin{itemize}
	\item $a>0, \varphi=1$: all funding strategies are equivalent;
	\item $a>0,\ a+bT>0, \varphi<1$: Term funding;
	\item $a>0,\ a+bT<0, \varphi<1$: Shortest possible.
\end{itemize}
\end{myProp}
\begin{proof}{\it Sketch.}
The first result is obvious by no-arbitrage: as there is no bid-ask spread no strategy can dominate.

When there is a bid ask spread and rolls are costly term funding dominates.

When there is a bid ask spread and observed negative zero rates then rolls become beneficial and the optimal strategy is to have as many rolls as possible.
\end{proof}

\begin{figure}[htbp]
	\centering
		\includegraphics[width=0.90\textwidth,clip,trim=0 0 0 0]{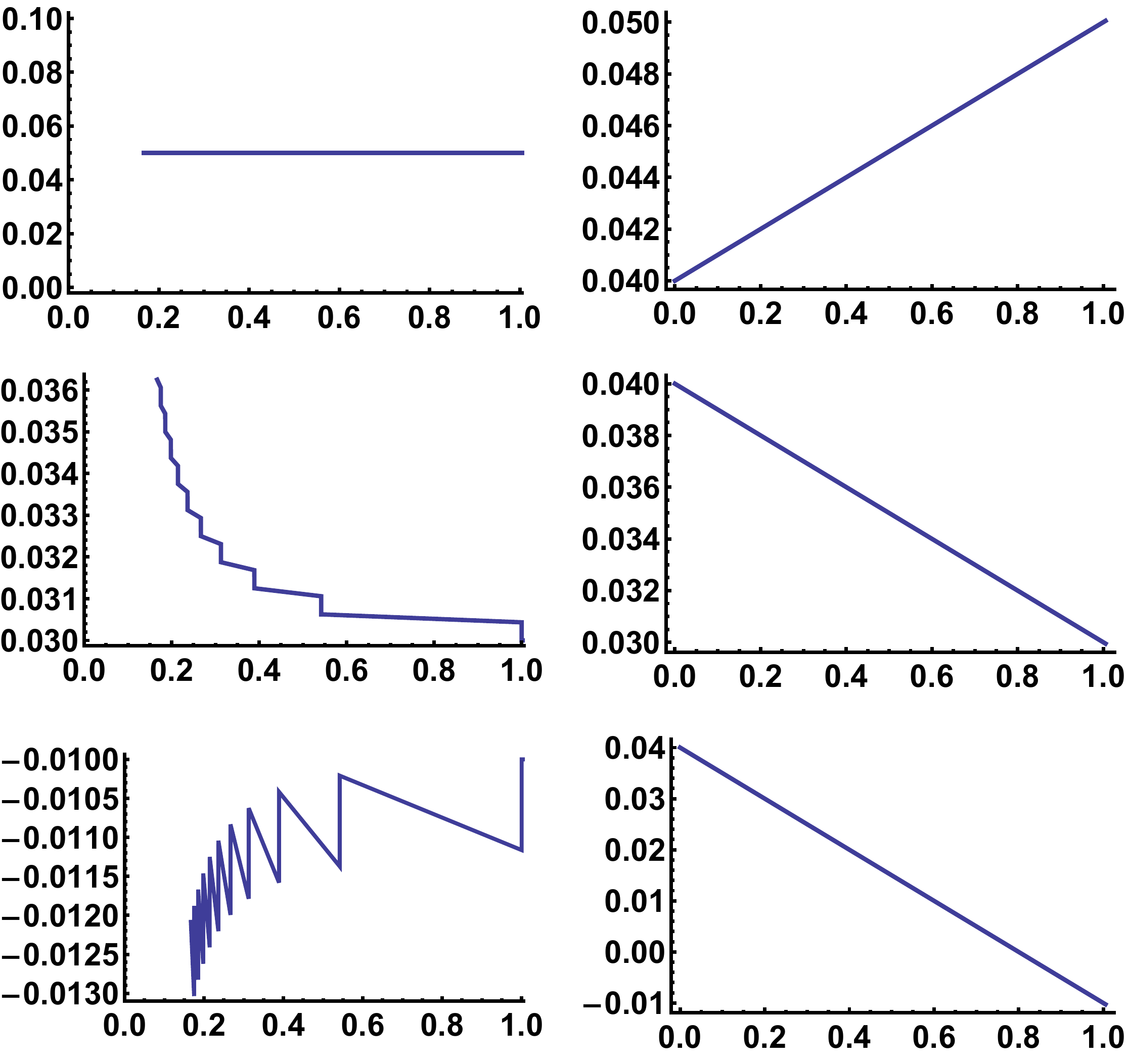}
	\caption{Left panels show funding costs versus roll lengths for: $a>0, \varphi=1$; $a>0,\ a+bT>0, \varphi<1$; $a>0,\ a+bT<0, \varphi<1$.  The right panels show the corresponding continuously compounded yield curves.  Parameters were: $h=1$; $\Delta=1/12$; with $\varphi=1$ for the top panel, and $\varphi=0.75$ for the lower two panels.  Axes for all graphs are the same: horizontal axes are roll length; vertical axes are funding costs.}
	\label{f:optQ}
\end{figure}

Figure \ref{f:optQ} illustrates the three cases. Note that although this funding is optimal in \Q\ terms there is no guarantee that following this strategy will result in the lowest costs in \Pbb\ terms.  This is precisely the hedging caveat of \cite{Hull2011a}.

\FloatBarrier
\subsubsection{Optimal \Pbb\ Funding}
 
\Pbb\ Funding is where we solve the optimization under \Pbb\ measures.  Potentially there will be a different solution for each \Pbb\ measure, although the problem is unchanged some of the inputs are different (i.e. measures and filtrations).

The \Pbb\ measure has no unique definition, unlike the \Q\ measure.  We will investigate several comparing the resulting strategies, and the \Q-optimal strategy with the \PbbPI-optimal strategy.  \FbbPI\ is the Perfect Information filtration and \PbbPI\ the corresponding (set of) measure(s) where we know the future (a rather simple set of measures in fact).  Thus we can calculate the value of perfect information (VPI) and state that the optimal strategy is the one with the lowest remaining VPI.  

The \Pbb\ situation is more complex than in the \Q\ case but there are some regularities that can be observed.  We start with the assumption of a \CONSTANT\ yield curve, i.e. future yield curves have exactly the same shape as today.
\begin{myProp}
Given a linear input yield curve $(a>0,b)$ the regulatory-optimal \Pbb-funding strategies where \Pbb=\CONSTANT\ yield curve, with horizon $h$ are:
\begin{itemize}
	\item $a>0, b>0, \varphi=1$: Shortest possible;
	\item $a>0, b>0, \varphi<1$: Neither the shortest nor Term funding are always optimal;
	\item $a>0,\ b<0$: Term funding;
\end{itemize}
\end{myProp}
\begin{proof}
{\it Sketch.}  
Given that rolls do not matter as the bid-ask spread is zero we can choose any roll.  As the future funding cost is the same as today, then with an upward-sloping yield curve the lowest cost is from the shortest roll.

When rolls are costly there is a playoff between moving up the yield curve for longer funding and fewer rolls against the cost of doing so.  Clearly the optimal strategy can move.

When the yield curve is downwards-sloping, and since $a>0$ and is always greater than zero rolls are always costly, hence the cheapest strategy is to have no rolls.  That is, term funding is optimal.
\end{proof}

\begin{figure}[htbp]
	\centering
		\includegraphics[width=0.90\textwidth,clip,trim=0 0 0 0]{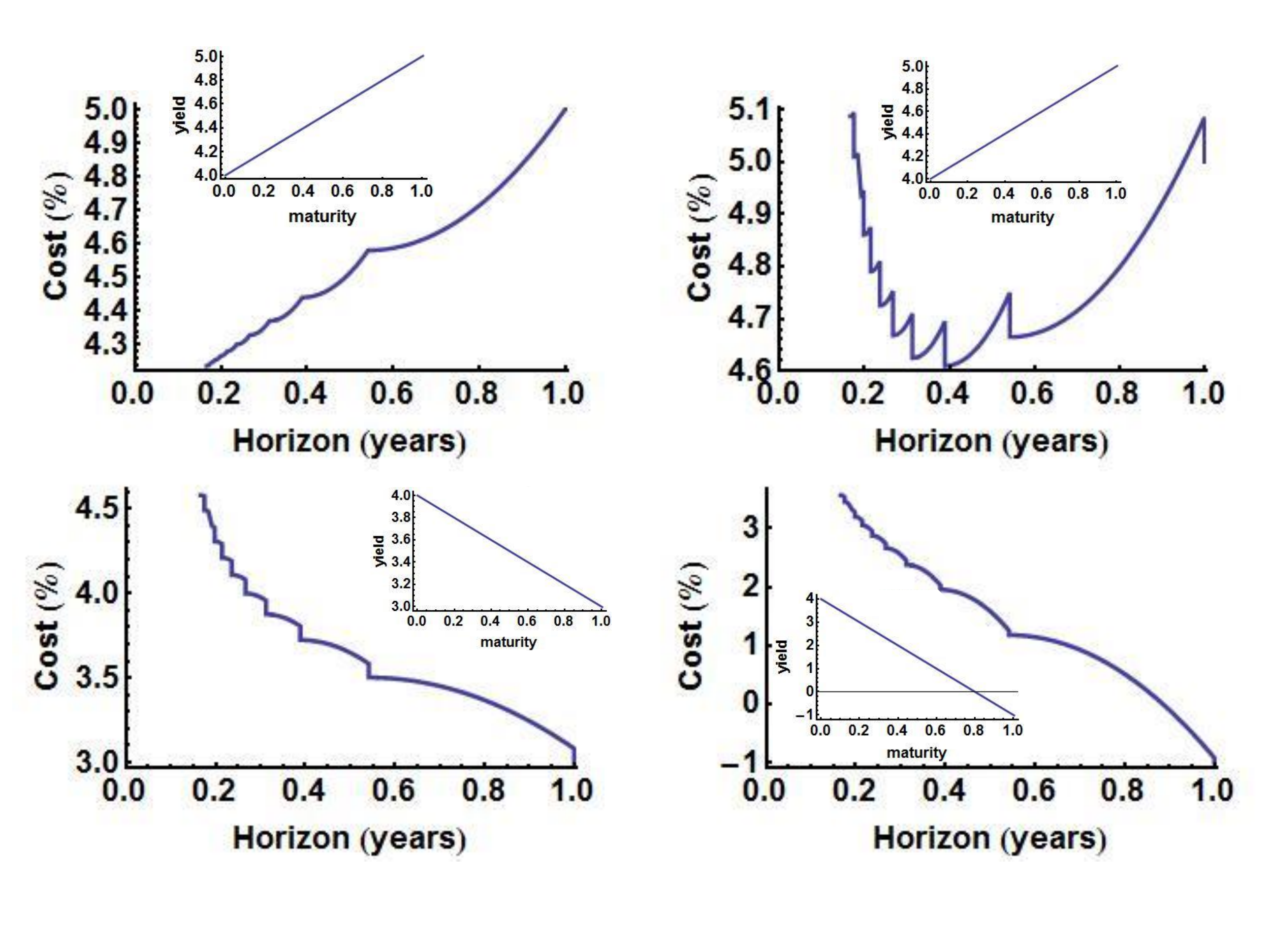}
	\caption{\Pbb\ Funding costs versus roll lengths for three of the same cases as Figure \ref{f:optQ}, with the addition of the upper right which has: $a>0, b>0, \varphi=0.75$.  Equivalence of panels is clockwise from top right.  Axes for all graphs are the same: horizontal axes are roll length; vertical axes are funding costs.}
	\label{f:optPconstant}
\end{figure}

Figure \ref{f:optPconstant} illustrates four cases, three are the same as Figure \ref{f:optQ} with the addition of the $a>0, b>0, \varphi<1$ case.  In this extra case the optimal strategy is intermediate between shortest possible and Term funding.  

Strikingly, the \Q-optimal and \Pbb-\CONSTANT-optimal strategies are almost opposite. The most realistic case is perhaps where there is a bid-ask spread and then the \Pbb-\CONSTANT-optimal strategy is intermediate, which is not even present within the \Q-optimal strategies.  

The \CONSTANT\ case is not as restrictive as it may appear.  What we are effectively doing is comparing the movement of the baseline with the differential of the yield curve.  In the second case the yield curve is steeper than the baseline movement, whilst in the third the baseline movement is steeper than the yield curve.  Thus the key to practical strategies is short rate momentum versus yield curve steepness.

However, we have not yet examined how these optimal strategies would hold up practically.  For this we turn to the Value of Perfect Information.

\FloatBarrier
\subsection{\Pbb\ Identification and Perfect Information Benchmark\label{s:Pid}}

We can only identify \Pbb\ on a statistical basis.  We want to obtain lower funding costs that are materially lower than \Q-funding on average and that are statistically significant out-of-sample.  The best we can do is if we know the future and this is our perfect information benchmark (often referred to as Expected Value of Perfect Information, or EVPI.  EVPI is a standard metric from the stochastic programing literature \cite{Birge2011a} and originally in \cite{Raiffa1961a}.  EVPI is the payment which makes a decision maker indifferent between receiving the payment and being given complete and accurate knowledge of the future (with respect to the decision).
\begin{figure}[htbp]
	\centering
		\includegraphics[width=1.00\textwidth]{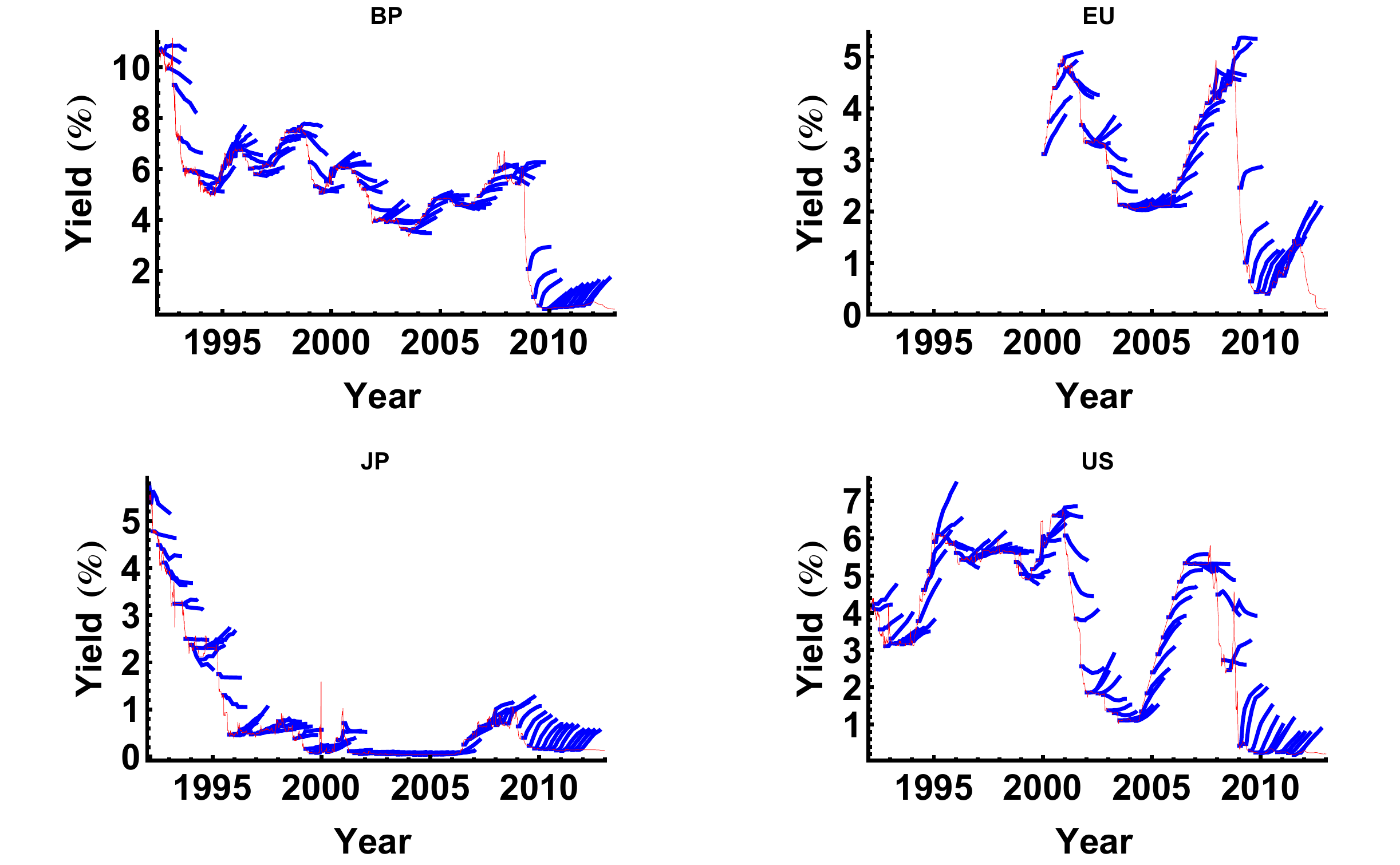}
	\caption{Libor curves (short blue lines; continuously compounding) for four major currencies compared to 1M Libor rate (continuous line; continuously compounding).}
	\label{f:4whisk}
\end{figure}
Figure \ref{f:4whisk} shows examples of funding (i.e. Libor) curves out to 1Y used in this example for the four major currencies considered (GBP, EUR, JPY, USD). 

In our \Pbb\ identification we assume that the shape of the yield curve does not change but that it undergoes parallel shifts.  We aim to predict these parallel shifts using a momentum strategy, specifically an exponentially weighted moving average (EWMA) filter (i.e. a trivial Kalman filter).  We add some refinements:
\begin{itemize}
	\item Limit gradient of short end movement so that projected base rates do not go negative
	\item Apply a threshold $\theta$ to the calculated gradient 
	\item Weight the gradient  by $\omega$, to between 0\%\ and 100\%\ of its calibrated value
	\item The projected short end of the curve cannot go negative
\end{itemize}
Thus we have three parameters to calibrate: $\{\lambda, \theta, \omega\}$.  $\lambda$ is the decay constant for the EWMA filter.  

\subsection{Results}

Calibration solves the problem below:
\ben
\max_{\Lambda\in \Pi} \sum_{c\in\{\text{BP,EU,JP,US}\}}\left(  \cav_{t,{\rm opt}}^{\Q}(h;c) -\cav_{t,{\rm opt}}^{\Pbb}(h;c,\Lambda) \right) /4\label{e:cal}
\een
where the ``/4'' is because we average over four currencies $c$, which are now visible as parameters of the costs $C$. Also:

$\Lambda = \{\lambda,\theta,\omega\}$ is the set of parameters. $\lambda \in[0,10\ \text{years}]$ is the EWMA decay constant (units of time); $\theta\in[0,1]$ is the gradient threshold (no units); and $\omega\in[0,1]$ is the gradient scaling (no units).  Parameter ranges over which to optimize are reasonable choices in the authors' opinions.

$c$ is a currency in the set of currencies \{\text{BP,EU,JP,US}\}.

$\Pi$ is the multi-dimensional space of feasible parameter values.

$h$ is the horizon, here one year.

Equation \ref{e:cal} is a maximization over a nested optimization, see Equations \ref{e:opt}, \ref{e:cav} for details of $\cav$.  In the inner optimization $\alpha$ is chosen.  In the outer optimization (Equation \ref{e:cal}) the parameters are chosen.  Note that \Q\ and \Pbb\ replace \ddag\ in Equations \ref{e:opt} and \ref{e:cav} as we are solving for the parameters which improve the \Pbb\ solution the most relative to the \Q\ solution.  We have also made the parameters in the \Pbb\ optimization explicit, ``$;\Lambda$''.  It is also clear that to understand whether there is a real improvement, or not, we must test the resulting calibration out of sample.

We calibrated model parameters using the first five years of data by maximizing the average funding cost improvement relative to a \Q-optimal strategy (i.e. Equation \ref{e:cal}), see Table \ref{t:calib}.  The calibration was done jointly over all the currencies because the main effects are due to economic cycles.   The model was then tested out-of-sample on the remaining data by comparing with the \Q-optimal strategy and the EVPI strategy.

\begin{table}[htbp]
	\centering
		\begin{tabular}{lcc}
		 & Parameter & Value \\ \hline
	\multirow{3}{*}{Setup}	& Horizon & 1Y \\
		&Minimum regulatory buffer & 1M \\
		&Bid-ask ($\varphi$) & 0.75 \\ \hline
	 \multirow{4}{*}{Calibration} 
	 & Calibration length & 5Y \\
	 & EWMA decay ($\lambda$) & 90D \\ 
	 & Gradient threshold ($\theta$) & 0.005 per day \\
	 & Gradient scaling ($\omega$) & 0.3 \\ \hline
		\end{tabular}
	\caption{Setup, and calibration for EWMA and additional parameters.  Calibration was combined over all the currencies.}
	\label{t:calib}
\end{table}

\begin{figure}[htbp]
	\centering
		\includegraphics[width=0.90\textwidth]{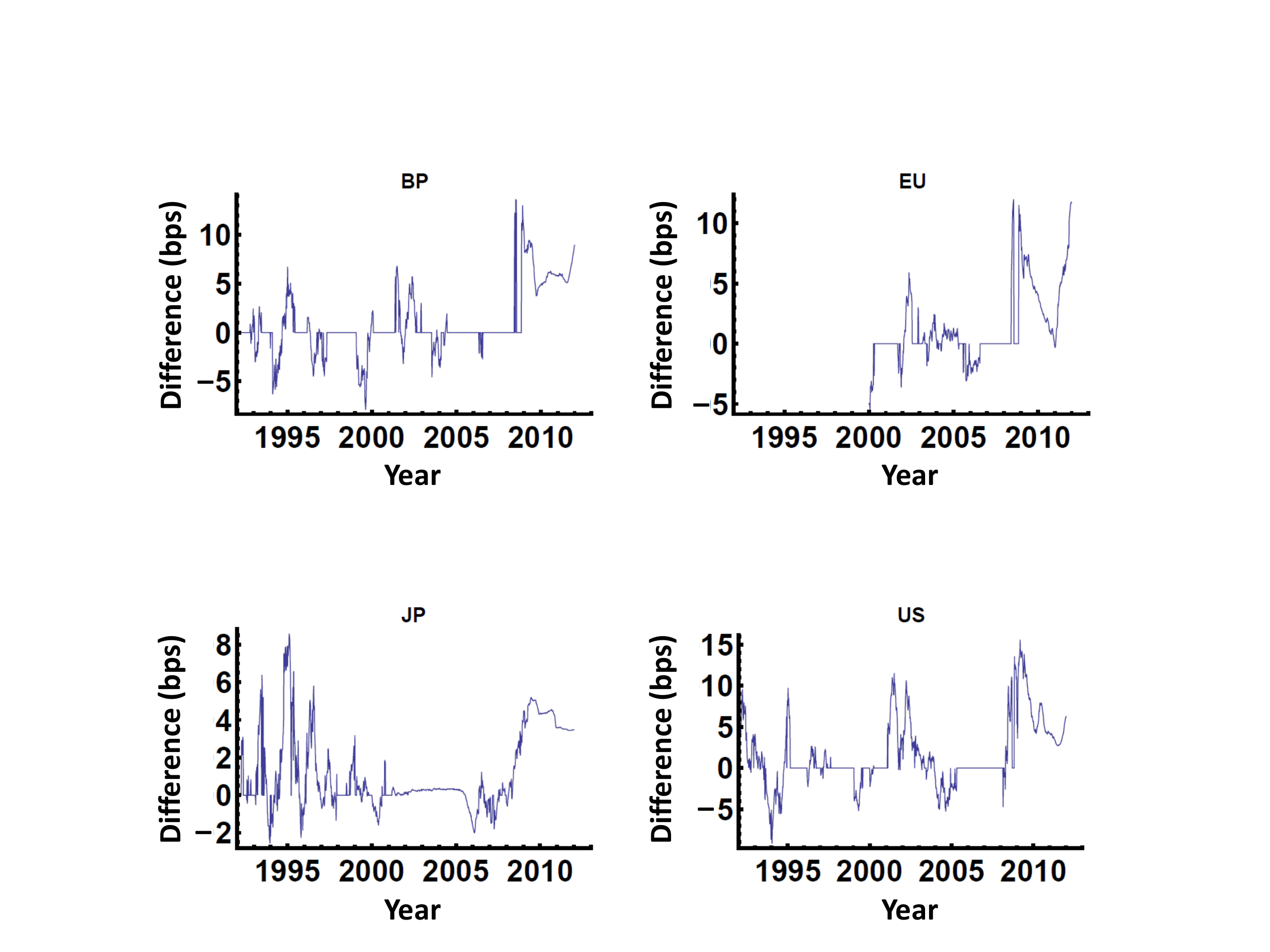}
	\caption{\Q-funding average undiscounted cost relative to EWMA-based funding cost (0.01 on vertical scale = 100bps).}
	\label{f:QvsEWMA}
\end{figure}
\begin{figure}[htbp]
	\centering
		\includegraphics[width=0.90\textwidth]{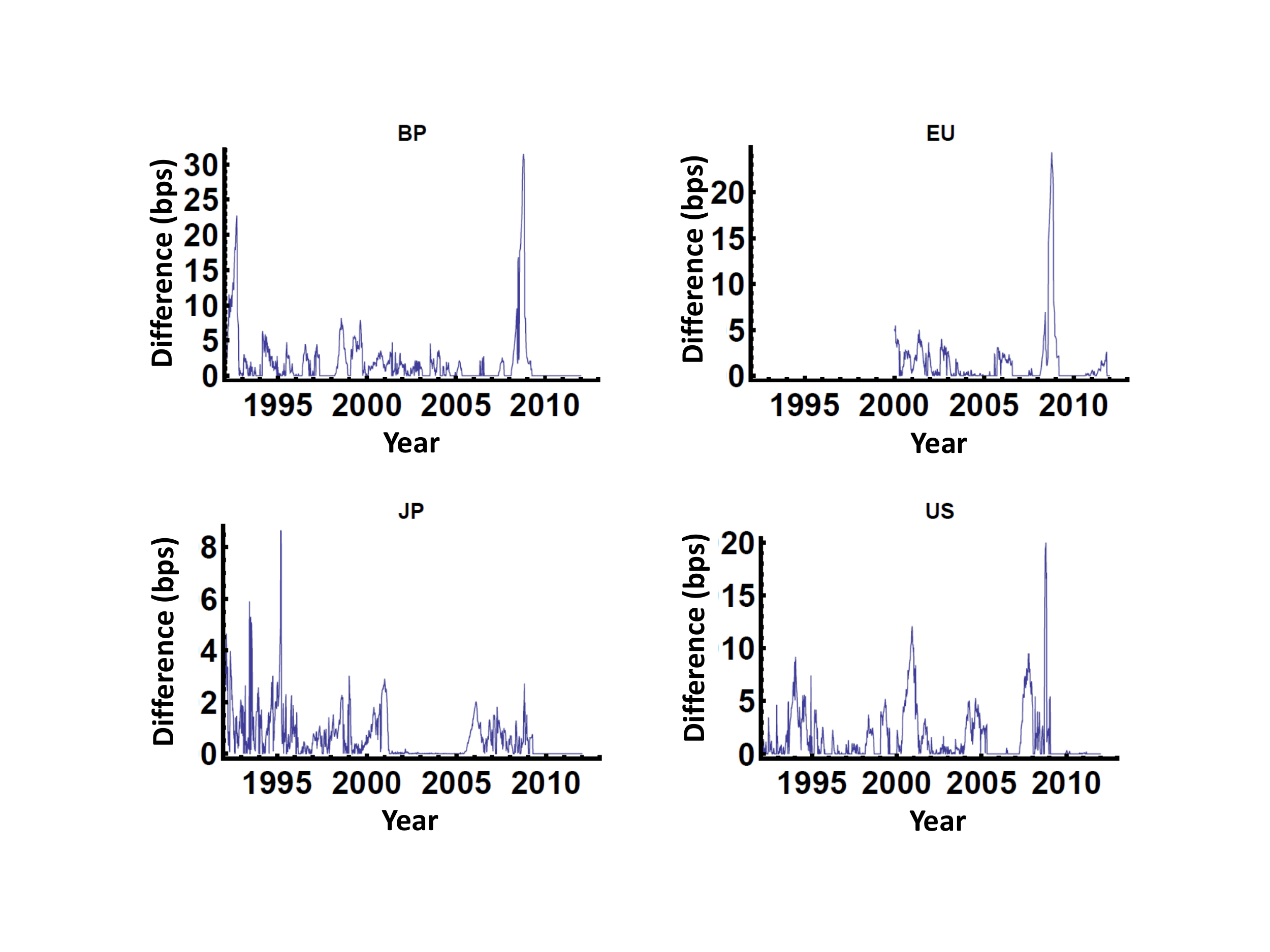}
	\caption{EWMA-based average undiscounted cost relative to perfect information funding cost (0.01 on vertical scale = 100bps).}
	\label{f:EWMAvsVPI}
\end{figure}
\begin{figure}[htbp]
	\centering
		\includegraphics[width=0.90\textwidth]{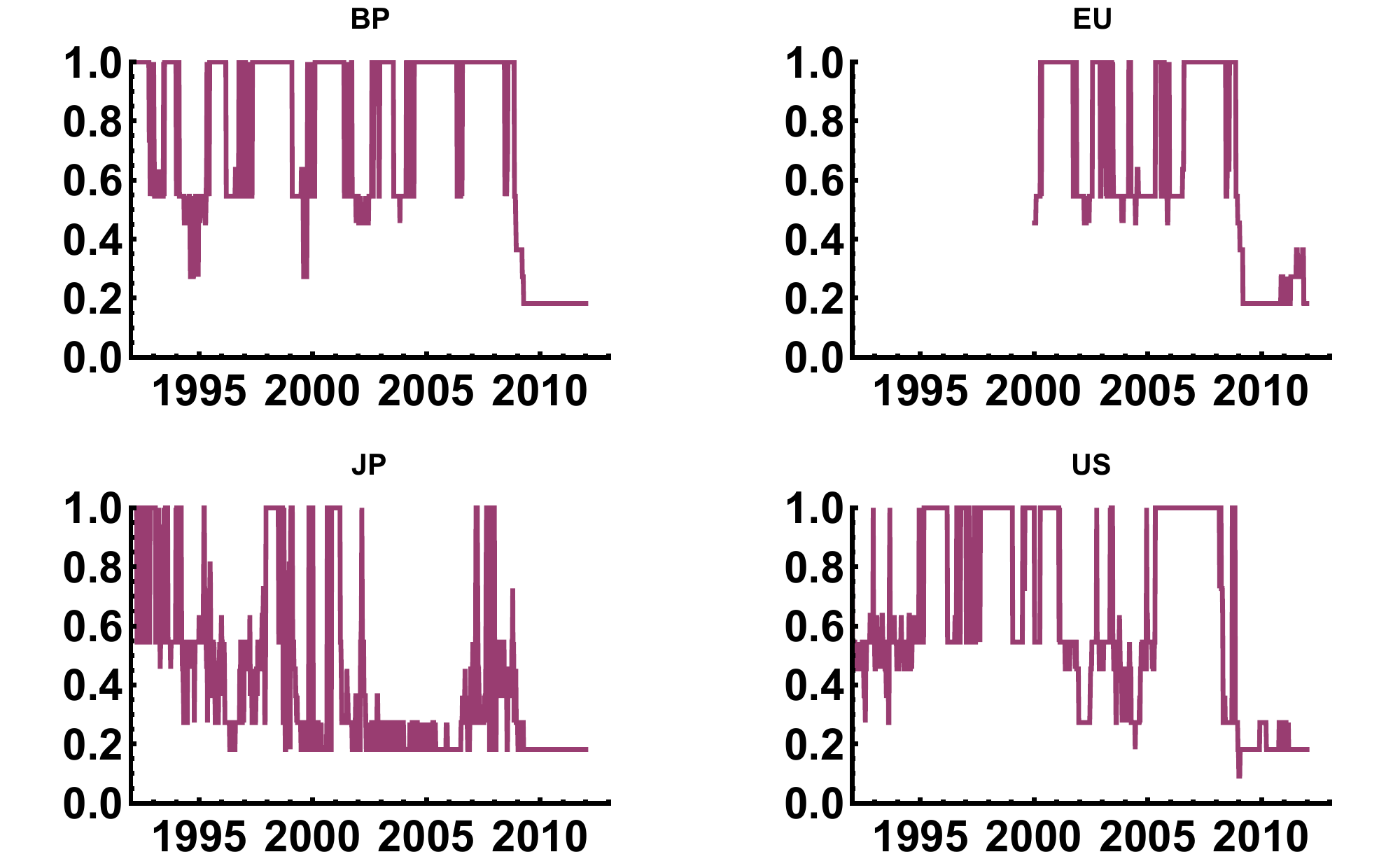}
	\caption{Optimal EWMA-based roll choice (years, on vertical scale). \Q-based roll choice is always term funding, i.e. here 1Y.}
	\label{f:EWMAroll}
\end{figure}

\begin{table}[htbp]
	\centering
		\begin{tabular}{ccccc}
		 & BP & EU & JP & US \\ \hline
\Q\ vs EWMA (average, bps) & 13 & 22 & 10 & 19 \\
T-Test p-level& E-25 & E-27 & E-45 & E-37 \\
EWMA vs EVPI (average, bps) & 16 & 15 & 4 & 15 \\ \hline
\%-efficient & 44\% & 59\% & 71\% & 56\%
		\end{tabular}
		\caption{Out-of-sample results for different funding strategies and their statistical significance (p-value).  Efficiency is defined as (\Q\ - EWMA) / (\Q\ - EVPI).}
	\label{t:res}
\end{table}

Figure \ref{f:QvsEWMA} shows optimal \Q\ funding costs versus our EWMA model (\Pbb\ optimal funding).  We see that following an optimal \Q\ funding strategy is always more expensive (positive cost differences) since 2008 for all currencies.  Prior to 2008 the picture is more mixed but there is a preponderance of positive costs for EU,JP,US whereas for BP the picture pre-2008 seems balanced.

Figure \ref{f:EWMAvsVPI} shows our EWMA model funding (\Pbb\ optimal funding) costs versus funding when the future is known (i.e. Value of Perfect Information).  For BP,EU,US our strategy is usually bounded by 0.005 (i.e. 50bps) and for JP it is 0.002 (or 20bps).  That is we are within 50bps (or 20bps for JP) of a perfect performance.  Figure \ref{f:QvsEWMA} by contrast shows that \Q\ optimal funding can be multiples of this bound worse than our strategy.  There are a few short spikes where a change of market conditions has not been picked up by our strategy quickly enough.  One or two spikes, per currency, over 20 years is an encouraging performance.

Figure \ref{f:EWMAroll} shows the EWMA model (\Pbb\ optimal funding) roll choices.  Generally the strategy is bang-bang\footnote{``bang-bang'' is a standard term in optimal control theory \cite{Craven1998a}.}, with dis-continuous changes in the control value (here the roll length).  It is noticeable that roll choices are even shorter when the model detects a significant mismatch between market-observed yield curves and market behavior.  This occurs when rates are low and the yield curve is upwards sloping or flat.  Recall that \Q\ always chooses term funding when there is a bid-ask spread, as here.  

In all Figures \ref{f:QvsEWMA}, \ref{f:EWMAvsVPI}, \ref{f:EWMAroll}, the first five years are calibration and the rest are performance.

The relative funding cost figures show that the major information missing is the exact timing of the start of the late-2008 drop in rates for GBP, EUR and USD.  For JPY the drop started gradually so the EWMA-based setup adapted smoothly.  Notice that the funding cost increase is a relatively sharp peak so we see that the EWMA adapted quickly to the new setup.  An additional point is that there is no equivalent cost peak relative to perfect information when rates approached zero.  This is easy to explain as we built-in knowledge that the short end of the funding (yield) curve could not go past zero.

Table \ref{t:res} summarizes the out-of-sample results.  We see that EWMA-based optimal funding is between 10bps and 22bps better on average relative to \Q-optimal funding.  However, EWMA-based funding still lacks 4bps to 16bps relative to perfect information.  We see that EWMA-based funding, for the example setup, achieves 44\%\ to 71\%\ efficiency.  Effeiciency is defined as achieved average improvement relative to \Q\ funding relative to perfect information.  Statistically these results are highly significant with p-values of at least 10E-20.

\FloatBarrier
\section{Discussion}

On a derivatives desk the question of funding costs is normally answered by "ask Treasury".  Here we take the view from Treasury and integrate regulatory constraints into an optimal funding problem.  We have shown theoretically that \Q-optimal and \Pbb-optimal strategies are radically different.  Practical funding optimization relies on identification of an appropriate \Pbb-measure on statistical grounds.  We have shown that this is possible and that quite simple strategies based on EWMA (i.e. momentum) are statistical, and practical, improvements on hedged funding.  These strategies achieve 44\%\ to 71\%\ efficency when compared to perfect information.

Prospective Prudent Valuation regulations \cite{EBA-CP-2013-28,EBA-CP-2013-28-FAQ} may see the funding strategies proposed here as business models.  Thus they would not be applicable for use in pricing trades for capital purposes.  

We have limited ourselves to a one-year funding horizon because that is how far out deposits are available.  We used these as proxy funding costs for a typical bank.  Future work could go out further using bond curves or CDS.  There are issues with both data sources. Bond prices are generally only visible in the secondary market which is not the one that the issuing bank has to optimize against.  CDS are unfunded instruments, have only been available for ten to fifteen years, they can have significant bond CDS bases.  These are outside the scope of this paper.

We consider funding costs in terms of expectations, i.e. very simple utility functions.  More complex utility functions could be introduced considering VAR or Expected Shortfall \cite{BCBS-219}.  However, these may not be appropriate because the buffer itself acts against downside risk so introducing further downside risk measures appears superfluous.  Other more complex utility functions would require justification.  Utility functions are equivalent to strategy choices, they cannot be used to compare different choices.  At executive level the utility function should be compared with the risk appetite of the bank.

Our yield curve prediction based on EWMA and our myopic optimization can be improved in many ways.  A Kalman filter, or machine learning approach could be applied to prediction.  The optimization could be moved from myopic to multi-stage stochastic optimization.  However, even at this early stage we achieve significant practical improvements, and these are at low computational cost.  We leave these developments for future research.

It may appear that our \Q-funding setup is unfair in that hedged funding is assumed to be possible in the optimization, but must then play out in the physical measure for its actual cost.  This is not a difficulty because (myopic) optimal \Q-funding is always term funding (when the yield curve is positive, as it is in practice).  Thus the anticipated funding cost is achieved.

Practically one may argue that taking downwards-sloping yield curves for funding curves is infeasible because funding at longer tenors may be volume-limited.  We leave consideration of market impact for future research but note that downwards-sloping curves are rare so we do not expect them to change our statistical conclusions.

This paper highlights the importance of regulatory liquidity constraints on funding, specifically the minimum liquidity buffer.  It also sets up the funding cost problem from the Treasury point of view and shows how it can be optimized using \Q and \Pbb\ points of view.  Demonstration results are significant both statistically and practically.

\bibliographystyle{alpha}
\bibliography{kenyon_general}
\end{document}